\RequirePackage{fix-cm}
\RequirePackage{amsmath,amssymb,amsfonts}
\RequirePackage{mathtools}
\RequirePackage{graphicx}
\RequirePackage{amsmath}
\RequirePackage{textcomp}
\RequirePackage{algorithm}
\RequirePackage{upgreek}
\RequirePackage{booktabs}
\RequirePackage{threeparttable}
\documentclass[smallextended]{svjour3}       
\smartqed  
\usepackage{graphicx}
\begin{document}

\title{A class of quantum synchronizable codes of length $2^n$ over ${\rm F}_q$ 
\thanks{This work was supported by Graduate Innovation Project of China University of Petroleum (East China) (YCX2021137).}
%
}

\author{Shiwen Sun\and Tongjiang Yan\and Yuhua Sun\and Tao Wang\and Xueting Wang}


\institute{Shiwen Sun \at
              \email{winner\_ssw@163.com}           
           \and
           Tongjiang Yan \at
           \email{yantoji@163.com} 
           \and
           Yuhua Sun \at
           \email{sunyuhua\_1@163.com} 
           \and
           Tao Wang \at
           \email{1409010215@s.upc.edu.cn} 
           \and
           Xueting Wang \at
           \email{wxt961425988@163.com}
}

\date{Received: date / Accepted: date}

\maketitle

\begin{abstract}
Quantum synchronizable codes (QSCs) are special quantum error-correcting codes which can be used to correct the effects of quantum noise on qubits and misalignment in block synchronization. In this paper, a new class of quantum synchronizable codes of length $2^n$ are constructed by using the cyclotomic cosets, whose synchronization capabilities always reach the upper bound $2^n$. 
\end{abstract}
\section{Introduction}\label{intro}
As a key technology for future communication security, quantum communication is bound to be commercially available on a large scale to provide reliable security for the development of the information society. In the process of quantum computing, the qubits of a quantum computer are not isolated, so they inevitably interact with the external environment which destroys the coherence between the qubits and leads to quantum decoherence. Using the quantum error-correcting codes is an important method to ensure the integrity and reliability of information in the process of quantum transmission. 

Quantum error-correcting codes are used to prevent and correct quantum errors by introducing redundant information in an appropriate way in order to improve the resistance of information to interference. In 1995, Shor \cite{Shor P W}gave a quantum error-correcting scheme which aims to reduce the effects of decoherence for information stored in quantum memory. In 1996, Calderbank et al. \cite{Steane A M,Calderbank S}proposed the quantum CSS construction method, which made it possible to construct quantum error-correcting codes by classical error-correcting codes with some special properties. Besides they established the connection between quantum error-correcting codes and classical error-correcting codes. Then, they systematically established a mathematical model of quantum error-correcting codes and gave an effective mathematical method to construct quantum error-correcting codes from classical error-correcting codes from ${\rm F}_2$ to ${\rm F}_4$, by using the characteristic theory of finite exchange groups, so as to construct some good quantum error-correcting codes \cite{Calderbank R}.

Quantum synchronizable codes are special quantum error-correcting codes that can correct both the effect of quantum noise on qubits and misalignment in block synchronization. In 2013, Fujiwara \cite{FujiBsf} proposed a construction method to generate a QSC by two cyclic codes $C_1$ and $C_2$ satisfying the chain condition $C_1^\perp\subseteq C_1\subseteq C_2$. Based on it, the QSC can not only correct at least up to $\lfloor\frac{d_2-1}{2}\rfloor$ bit errors and at least up to $\lfloor\frac{d_1-1}{2}\rfloor$ phases errors, but also misalignment up to the left by $c_l$ qubits and up to the right by $c_r$ qubits, where $c_l$ and $c_r$ are both nonnegtive integers. In addition, the method has been proved by Fujiwara et al. which makes the tolerable magnitude of misalignment reach the order of $f(x)=\frac{g_1(x)}{g_2(x)}$, where $g_1(x)$ and $g_2(x)$ are the generator polynomials of $C_1$ and $C_2$ respectively \cite{FujiTW}. Based on the above theory, some QSCs have been constructed recent years. In 2014, Xie et al. constructed QSCs from quadratic residue and their supercodes \cite{XieYFu}. Then Li et al. proposed a new method of QSC construction by using the cyclotomic classes of order four in 2019 \cite{LiZL}. Afterwards, Shi et al. \cite{ShiQX} obtained some QSCs from the Whiteman's generalized cyclotomy which is a powerful tool to construct many combinatorial designs in 2020. Recently, some QSCs from cyclotomic classes of order two are constructed by Wang \cite{WangT}.

In 2012, Bakshi et al. proposed a class of constacyclic codes over a finite field and obtained all the self-dual negacyclic codes of length $2^n$ over ${\rm F}_q$ \cite{G.K. Bakshi}. Inspired by this article and the theory proposed by Fujiwara et al., we construct two classes of QSCs of length $2^n$ from the cyclic codes obtained by $q$-cyclotomic cosets. These QSCs can possess good error-correcting capability towards bit errors and phase errors and their misalignment tolerance can reach $2^n$ which is the upper limit.

The rest of the paper is organized as follows. Some definitions and notations on cyclic codes and cyclotomic cosets are introduced in Section \ref{sec:1}. In Sections \ref{sec:3} and \ref{sec:4}, we deduce some dual-containing codes of length $2^n$ over ${\rm F}_q$ and some augmented cyclic codes respectively. Besides, we derive some necessary lemmas to get our main results. Then two classes of QSCs are constructed in Section \ref{sec:5}.

\section{Preliminaries}
\label{sec:1}
Let ${\rm F}_q$ be a finite field, where $q$ is a power of an odd prime. In this section, we will introduce some properties of cyclic codes and cyclotomic cosets that are needed in our consideration later.
\subsection{Cyclic codes}
A linear code $C$ of length $N$ and dimension $k$ over ${\rm F}_q$ is often called a $[N,k]$-linear code. If the minimum distance $d$ of $C$ is known, it is also sometimes referred to as an $[N,k,d]$-linear code.
Besides, a linear code $C$ is called a cyclic code if $\{c_0,c_1,\cdots,c_{N-1}\}$ is a codeword of $C$ and $\{c_{N-1},c_0,,\cdots,c_{N-2}\}$ is also  a codeword of $C$. In order to convert the combinatorial structure of cyclic codes into an algebraic one, we consider the following correspondence\cite{San Ling1}
\begin{equation*}
	\pi:{\rm F}_q^N\to {\rm F}_q[x]/(x^N-1),\ (a_0,a_1,\cdots,a_{N-1})\mapsto a_0+a_1x+\cdots+a_{N-1}x^{N-1}.
\end{equation*}
Then $\pi$ is an ${\rm F}_q$-linear transformation of vector spaces over ${\rm F}_q$. Now we can identify a codeword $(c_0,c_1,\cdots,c_{N-1})$ with the polynomial $c(x)=\sum_{i=0}^{N-1}c_ix^i$. From the definition of $\pi$, we know that a nonempty subset C of ${\rm F}_q^N$ is a cyclic code if and only if $\pi(C)$ is an ideal of ${\rm F}_q[x]/(x^N-1)$. Besides, if $I$ is a nonzero ideal in ${\rm F}_q[x]/(x^N-1)$ and $g(x)$ is a nonzero monic polynomial of the least degree in $I$, then $g(x)$ is a generator of $I$ and divides $x^N-1$. And $h(x)=\frac{x^N-1}{g(x)}$ is called the parity check polynomial of $C$. The dual code of $C$ is defined by $C^\perp=\{c'\in {\rm F}_q^N|(c',c)=0$ for all $ c\in C\}$, where $(c',c)$ is the Euclidean inner product of $c'$ and $c$. The generator polynomial of $C^\perp$ is $\stackrel{\thicksim}{h}(x)$, where 	$\stackrel{\thicksim}{h}(x)=h(0)^{-1}x^{deg(h(x))}h(x^{-1})$ is the reciprocal polynomial of $h(x)$.

\begin{lemma}\label{BCH}\cite{Mac Willianms}
	Let $C$ be a cyclic code of length $N$ over ${\rm F}_q$. Let $\alpha$ be a primitive $Nth$ root of unity in an extension field of ${\rm F}_q$. Assume the set consisting of the roots of the generator polynomial of $C$ includes the set $\{\alpha^i|i_1\le i\le i_1+d-2\}$. Then the minimum distance of $C$ is at least $d$.
\end{lemma}

\subsection{Cyclotomic cosets}
\label{sec:2}
For any $s\geq 0$, let $C_s=\{s,sq,sq^2,\cdots,sq^{m_s-1}\}$ denote the $q$-cyclotomic coset containing $s$ modulo $2^n$, where $m_s$ is the least positive integer such that $sq^{m_s}\equiv s({\rm mod} \ 2^n)$. 
\begin{definition}\label{alpha_def}
	Let $\alpha$ be a primitive $2^n$-th root of unity in some extension field of ${\rm F}_q$. It is well known that \cite{G.K. Bakshi}
	\begin{equation}\label{M_s(x)}
		M_s(x)=\prod_{i\in C_s}(x-\alpha^i)	
	\end{equation}
	is the minimal polynomial of $\alpha^s$ over ${\rm F}_q$ and
	\begin{equation*}
		x^{2^n}-1=\prod_{s\in \omega} M_s(x)
	\end{equation*}
	gives the factorization of $x^{2^n}-1$ into irreducible factors over ${\rm F}_q$, where $\omega$ is a complete set of representatives from distinct $q$-cyclotomic cosets modulo $2^n$. 
\end{definition}

\begin{lemma}\label{lm1}\cite{G.K. Bakshi}
	Let n $\geq 3$ and $q=1+2^zc$, where $z\geq 2$ and $c$ is odd. All the distinct $q$-cyclotomic cosets modulo $2^n$ are given by $C_0$, $C_{2^{n-1}}$ and $C_{S2^{n-r}}$ for $2\leq r\leq n$ and $S$ runs over $S_r$ for each r, where
	\begin{equation}\label{Sr}
		S_r=
		\begin{cases}
			\{\pm 1,\pm 3,\cdots ,\pm 3^{(2^{z-2}-1)}\},\ z+1\leq r\leq n,\\
			\{\pm 1,\pm 3,\cdots ,\pm 3^{(2^{r-2}-1)}\},\ 2\leq r\leq z.			
		\end{cases}
	\end{equation}
\end{lemma}

\begin{definition}
	According to the definition of $M_s$, denote $s$ in a more specific form, i.e., $S2^{n-r}$, where $S, n, r$ are the same as in Lemma \ref{lm1}. Now we define $M_{S2^{n-r}}$ as follows 
	\begin{equation}\label{M_S2^{n-r}_def}
		M_{S2^{n-r}}(x)=\prod_{i\in C_{S2^{n-r}}}(x-\alpha^i).	
	\end{equation}
\end{definition}

\begin{lemma}\label{pro_order(q)}
	For any $s\geq 0$, let $C_s=\{s,sq,sq^2,\cdots,sq^{m_s-1}\}$ be the same meaning as in Lemma \ref{lm1} and $m_s$ be the least positive integer such that $sq^{m_s}\equiv s({\rm mod}\ 2^n)$. Then we have
\begin{equation}\label{order(q)}
	m_s=\begin{cases}
		2^{n-z},&n\geq z+1, \\
		1,&n\leq z.
	\end{cases}
\end{equation}
\end{lemma}

\begin{proof}
 Combining with Lemma \ref{lm1}, we only consider the $q$-cyclotomic cosets given by $C_{S2^{n-r}}$(where $r=n$) as an example, others can be proved similarly. 

\noindent(i)For $n\leq z$,

because of $q=1+2^zc$, we have $2^z\big |q-1$. When $n=z$, the conclusion obviously holds; when $n<z$, we have $2^n\big |2^z$, so it is easy to have $2^n\big |q-1$. Hence, ${\rm ord}(q)=1$, i.e. $m_s=1$.

\noindent(ii)For $n\geq z+1$, 

when $n=z+1$, because of $2^z\big |q-1$ and $q\equiv1({\rm mod}\ 4)$, we have  $2^{z+1}\big |q^2-1$. Next we prove our conclusion by mathematical induction. First, suppose $n=k\geq z+1$, $2^k\big |q^{2^{k-z}}-1$ is true. Second, we prove when $n=k+1$, $2^{k+1}\big |q^{2^{k+1-z}}-1$ is true and $k+1$ is the minimum integer. We know that $q^{2^{k+1-z}}-1=(q^{2^{k-z}}-1)(q^{2^{k-z}}+1)$. By the assumed conditions and  $2\big |q^{2^{k-z}}+1$, the integer-division relationship is obtained.Now let $s$ is the order of $q$, then we have $2^{k+1}\big |q^s-1$. 
And because $2^k\big |q^{2^{k-z}}-1$, we have $2^k\big |{\rm gcd}(q^s-1,q^{2^{k-z}}-1)$, i.e. $2^k\big |q^{{\rm gcd}(s,2^{k-z})}-1$, then $2^{k-z}\big |s$. Let $s=2^{k-z}s_1$. If $s_1=1$, we have  $2^{k+1}\big |q^{2^{k-z}}-1$, then $2^k\big |q^{2^{k-z-1}}-1$, which contradicts the hypothesis. So $s_1\geq 2$. Then $s=2^{k-z+1}$ and we obtain the value of $m_s$.\qed
\end{proof}

\section{Dual-containing codes of length $2^n$}
\label{sec:3}
In this section, we study some factorizations of cyclotomic polynomials to construct dual-containing codes.
\begin{lemma}\label{lm3}\cite{G.K. Bakshi}
	Let $n\geq 1$, $q=1+2^zc$, $z\geq 2$, and c is odd, i.e. $q\equiv 1\left({\rm mod}\ 4 \right)$. Then we have
	\begin{equation*}
		x^{2^n}-1=M_0(x)M_{2^{n-1}}(x)\prod_{r=2}^{n}\prod_{S\in S_r}M_{S2^{n-r}}(x),
	\end{equation*}
where $S_r$ is mentioned in Lemma \ref{lm1}.
\end{lemma}

\begin{lemma}\label{dual_containing}\cite{WuYueS}
	Let the notations be the same as above. Let C be a cyclic code of length $N$ over ${\rm F}_q$ with the generator polynomial $g(x)$. Then C contains its dual code if and only if
	\begin{equation*}
		g(x)=h_1(x)^{b_1}(\stackrel{\thicksim}{h}_1(x))^{c_1}h_2(x)^{b_2}(\stackrel{\thicksim}{h}_2(x))^{c_2}\cdots h_t(x)^{b_t}(\stackrel{\thicksim}{h}_t(x))^{c_t},
	\end{equation*}
	  where $b_j, c_j\in \{0,1\}$ and $b_j+c_j\leq 1, 1\leq j\leq t$.
\end{lemma}
From Lemma \ref{dual_containing}, we know that the cyclic code $C=\langle g(x)\rangle$ contains its dual code if and only if $g(x)$ has not any self-reciprocal irreducible factor over ${\rm F}_q$. 

Let $n\geq 1, q=1+2^zc, z\geq 2$, and $c$ is odd. Then let the polynomial 
\begin{equation}\label{generator_polynomials}
	g(x)=\prod_{r=2}^{n}\prod_{S^+\in S_r^+}\prod_{S^-\in S_r^-}(M_{S^+2^{n-r}}(x))^{\upepsilon_{S2^{n-r}}}(M_{S^-2^{n-r}}(x))^{\upvarepsilon_{S2^{n-r}}},
\end{equation}
where $S_r^+$ and $S_r^-$ are two sets which contains all the positive elements and negative elements in $S_r$ respectively, $\upepsilon_i$ and $\upvarepsilon_i$ are equal to 0 or 1, $\upepsilon_i+\upvarepsilon_i\le 1$ and $i=S2^{n-r}, S\in S_r$.
\begin{lemma}\label{gnerate_matrix_dual_containing}
	Let $C_a$ be the cyclic codes generated by $g(x)$ defined in \eqref{generator_polynomials}. Then $C_a^\perp\subseteq C_a$.
\end{lemma}
\begin{proof}
	By the definition of the reciprocal polynomial and $M_s(x)$ defined by \eqref{M_s(x)}, we obtain the reciprocal polynomial of $M_s(x)$ is
	\begin{equation*}
		\begin{aligned}
		\stackrel{\thicksim}{M}_s(x)
		=&(M_s(0))^{(-1)}x^{|C_s|}M(x^{-1})\\
		=&\prod_{i\in C_s}(-\alpha^i)^{-1}\prod_{i\in C_s}(1-x\alpha^i)\\
		=&\prod_{i\in C_s}(x-\alpha^{-i}),
	\end{aligned}
	\end{equation*}
	where $|C_s|$ denotes the element number in coset $C_s$ and $\alpha$ is given by Definition \ref{alpha_def}.
	According to the definition of $C_s$, $-1 \in C_{-1}$. Then 
	\begin{equation}\label{g_R(x)}
		\stackrel{\thicksim}{M}_s(x)=\prod_{i\in C_{-s}}(x-\alpha^i).
	\end{equation}
	From Lemma \ref{lm1}, $C_{-1}\neq C_1$. Hence, we obtain $C_s\neq C_{-s}$. Then $M_s(x)\neq \stackrel{\thicksim}{M}_s(x)$. By Lemma \ref{dual_containing}, the lemma is proved.\qed
\end{proof}


\begin{table}
	\caption{Some dual-containing cyclic codes}
	\label{tab:1}       
	\begin{threeparttable} 
	\begin{tabular}{llll}
		\hline\noalign{\smallskip}
		$g(x)$ & Codes & Duals & Optimal or almost optimal \cite{codetable}  \\
		\noalign{\smallskip}\hline\noalign{\smallskip}
		$g_\delta(x)(x-\alpha^2)$ & $[8,5,3]_5$ & $[8,3,4]_5$ & $C_a$ is optimal and $C_a^{\perp}$ is almost \\
		&	&	& optimal\\
		$g_\delta(x)(x-\alpha^3)(x-\alpha^6)(x-\alpha^7)$ & $[8,3,4]_5$ & $[8,5,2]_5$ & both almost optimal\\
		$g_\delta(x)g_\chi(x)$ & $[16,7,8]_{9}$ & $[16,9,6]_{9}$ & both almost optimal\\
		
		\noalign{\smallskip}\hline
	\end{tabular}
\begin{tablenotes}
	\footnotesize	
\item Note: $g_\chi(x)=(x-\alpha^3)(x-\alpha^9)(x-\alpha^{11})(x-\alpha^{13})(x-\alpha^8)(x-\alpha^{10})(x-\alpha^{12})(x-\alpha^{14}), $

$ \qquad \ \  g_\delta(x)=(x-\alpha)(x-\alpha^5).$

\end{tablenotes}
\end{threeparttable}
\end{table}

\section{Augmented cyclic codes}
\label{sec:4}
In this section, we deduce some augmented cyclic codes by the cyclic codes given in section \ref{sec:3}. 
\begin{lemma}\label{construct_augmented_c}
	Let the symbols be the same as before. Suppose two types of cyclic codes are
	\begin{equation*}
	\begin{aligned}
	C_a=\langle&\prod_{r=2}^{n-1}\prod_{S^+\in S_r^+}\prod_{S^-\in S_r^-}(M_{S^+2^{n-r}}(x))^{\upepsilon_{S2^{n-r}}}(M_{S^-2^{n-r}}(x))^{\upvarepsilon_{S2^{n-r}}}\\
	&\prod_{S^+\in S_r}(M_{S^+}(x))^{\upepsilon_1}\prod_{S^-\in S_r}(M_{S^-}(x))^{\upvarepsilon_1}\rangle
	\end{aligned}
	\end{equation*}
	 and 
	 \begin{equation*}
	 \begin{aligned}
	 C_b=\langle&\prod_{r=2}^{n-1}\prod_{S^+\in S_r^+}\prod_{S^-\in S_r^-}(M_{S^+2^{n-r}}(x))^{\phi_{S2^{n-r}}}(M_{S^-2^{n-r}}(x))^{\varphi_{S2^{n-r}}}\\
	 &\prod_{S^+\in S_r}(M'_{S^+}(x))^{\phi_1}\prod_{S^-\in S_r}(M'_{S^-}(x))^{\varphi_1}\rangle, 
	\end{aligned}
	 \end{equation*}
  where $S_r^+$ and $S_r^-$ are two sets which contain all the positive elements and negative elements in $S_r$ respectively, $M'_i(x)$ is the factor of $M_i$, $\upepsilon_i$, $\upvarepsilon_i$, $\phi_i$ and $\varphi_i$ are equal to 0 or 1 for each relevant $i$, $i=S2^{n-r}$, $S\in S_r$, and $\upepsilon_i+\upvarepsilon_i\le 1, \phi_i+\varphi_i\le 1$ and $\upepsilon_1+\upvarepsilon_1=1$. Besides, $\phi_i\le \upepsilon_i$ and $\varphi_i\le \upvarepsilon_i$. Then we have $C_a\subseteq C_b$. 
	
\end{lemma}

\begin{proof}
The result can be obtained easily by \eqref{generator_polynomials} and Lemma \ref{gnerate_matrix_dual_containing}.\qed
\end{proof}
\begin{table}
	\caption{Some augmented cyclic codes}
	\label{tab:2}       
	\begin{tabular}{lll}
		\hline\noalign{\smallskip}
		$C_a$ & $C_b$ & Optimal or almost optimal \cite{codetable}  \\
		\noalign{\smallskip}\hline\noalign{\smallskip}
		$[8,5,3]_5$ & $[8,7,2]_5$ & both optimal\\
		$[16,11,3]_5$ & $[16,15,2]_5$ & $C_a$ is almost optimal and $C_b$ is optimal\\
		$[16,7,8]_{9}$ & $[16,9,6]_{9}$ & both almost optimal\\
		
		\noalign{\smallskip}\hline
	\end{tabular}
\end{table}

\section{Quantum synchronizable codes from cyclotomic polynomials}
\label{sec:5}
Now, we construct some quantum synchronizable codes from the cyclic codes given in section \ref{sec:4}.
\subsection{Quantum synchronizable codes}
Formally speaking, if a coding scheme encodes $k$ logical qubits into $N$ physical qubits and corrects misalignment by up to $c_l$ qubits to the left and up to $c_r$ qubits to the right, we call it a $(c_l,c_r)-[[N,k]]$ quantum synchronizable code. The following gives the general construction of quantum synchronizable codes.
\begin{lemma}\label{quantum_construction}\cite{FujiBsf}
	Let $C_1$ be an $[N,k_1,d_1]_q$ cyclic code with generator $g_1(x)$ and $C_2$ be an $[N,k_2,d_2]_q$ cyclic code with generator $g_2(x)$ such that $k_1<k_2$. Assume that $C_1^\perp \subseteq C_1\subseteq C_2$ and $f(x)=\frac{g_1(x)}{g_2(x)}$. Then there exists a $(c_l,c_r)-[[N+c_l+c_r,2k_1-N]]_q$ quantum synchronizable codes that corrects at least up to $\lfloor \frac{d_2-1}{2}\rfloor$ bit errors and at least up to $\lfloor \frac{d_1-1}{2}\rfloor$ phase errors, where $c_l$ and $c_r$ are two nonnegative integers satisfying $c_l+c_r<{\rm ord}(f(x))$.
\end{lemma}

\subsection{Maximum misalignment tolerance}
In this subsection, let $C_{\pm i}$ denote the cyclotomic coset pairs, where $i=S2^{n-r},n\ge 3, 2\le r\le n$ and $S\in S_r$. For any $f(x)\big |x^{2^n}-1$. If ${\rm ord}(f(x))=2^n$, then the QSC has the maximum misalignment tolerance, where ${\rm ord}(f(x))$ is defined as follows. If $f(0)\neq 0$, then the order of $f(x)$ is the smallest integer $e$ satisfying $f(x)\big |x^e$. If $f(0)=0$, then $f(x)=x^\tau g(x)$, where $g(0)\neq 0$, then ${\rm ord}(f(x))={\rm ord}(g(x))$.

\begin{lemma}\label{ord(f(x))bound}\cite{ShiQX}
	Let $f(x)$ be defined as above. If there exists an $\alpha^i$ such that $f(\alpha^i)=0$, where $(i,2^n)=1$, then ${\rm ord}(f(x))=2^n$.
\end{lemma}

\begin{lemma}\label{BCH_bound}
	Suppose that $r=n$ and $z=n-1$. Let the cyclic codes  $C_a=\langle g_1(x)\rangle$ and $C_b=\langle g_2(x)\rangle$, where 
	$\hat{M}_S(x)\prod_{j=1}^{2^{n-3}}(x-\alpha^{2j}) \big | g_1(x), 
	  \overline{\hat{M}_S(x)}\prod_{j=1}^{2^{n-3}}(x-\alpha^{2j}) \big | g_2(x) \big | g_1(x)$,  $\hat{M}_S=\prod_{k=1}^{2^{n-3}}(x-\alpha^{2k-1})(x-\alpha^{e_k})$, and $\overline{\hat{M}_S(x)}=\hat{M}_S/((x-\alpha)(x-\alpha^{e_1}))$, $e_k$ is another element of the cyclotomic coset in which $2k-1$ is located for relevant $k$ and $\alpha$ is given by Definition \ref{alpha_def}. Then $C_a^\perp\subseteq C_a\subseteq C_b$, and the minimum distance $d_1$ of $C_a$ is at least $2^{n-2}+1$, the minimum distance $d_2$ of $C_b$ is at least $2^{n-2}$. 	
\end{lemma}

\begin{proof}
In this paper, we define the odd-even relation and order relation of elements in $Z_{2^n}$ as them in the integer set.
According to the assumptions, the elements in cyclotomic cosets $C_{S2^{n-r}}$ correspond to the elements in $Z_{2^n}^*$. Recall that the order of $q$ is $2^{r-z}$ by \eqref{pro_order(q)}, so there are two elements $s,sq$ in each cyclotomic cosets. By \eqref{Sr}, we have 
\begin{equation*}
	M_S(x)=M_{\pm 1}(x)M_{\pm 3}\cdots M_{\pm 3^{2^{n-3}-1}}(x).
\end{equation*}
Then we get $2^{n-3}$ odd numbers by taking the smallest odd number in each cyclotomic coset pair $C_{\pm s}$ whose elements are in $\{0,1,\cdots,2^n-1\}$ by module $2^n$. So there exists at most $2^{n-3}$ consecutive odd numbers if and only if the next consecutive odd number is not in a previously taken cyclotomic coset pairs. Now we prove this fact. Suppose $v=2k-1$, where $k=1,2,\dots,2^{n-3}$. Then $-v=2^n-v({\rm mod}\ 2^n)\ge 2^n-1({\rm mod}\ 2^n)$, and 
\begin{equation*}
	\begin{aligned}
-vq	=&(2^n-v)q \left({\rm mod}\ 2^n \right)\\
	=&(2(2^{n-1}-k)+1)(2^{n-1}c+1)({\rm mod}\ 2^n)\\
	=&2^{n-1}-2k+1({\rm mod}\ 2^n)\ge 2^{n-2}+1. 
	\end{aligned}
\end{equation*}
So odd numbers $1,3,\cdots,2^{n-2}-1$ are not in $C_{-v}$. In order to satisfy the condition of Lemma \ref{BCH}, now we find some evens . Since ${\rm ord}(q)=1$ when $r\neq n$, there is only one element in each cyclotomic coset. That means we must be able to get different evens in different cyclotomic cosets. Then we get at least $2^{n-2}$ consecutive numbers. By Lemma \ref{BCH},  the cyclic code $C_a$ has the minimum distance $d_1$ which is at least $2^{n-2}+1$. Notice that the reciprocal polynomial of $\prod_{i\in C_s}(x-\alpha^i)$ is $\prod_{i\in C_{-s}}(x-\alpha^i)$ and $g_2(x)\big |g_1(x)$, hence $C_a^\perp\subseteq C_a\subseteq C_b$. The minimum distance $d_2$ of the cyclic code $C_b$ can be obtained similarly. \qed

\end{proof}

\begin{example}
	Let $n=5,z=4$, and $q=1+2^zc=17$. Then $C_S$ can be devided into eight cyclotomic cosets since $S$ can be taken as $\pm 1,\pm 3,\pm 3^2$ and $\pm 3^3$ in this case. Then  $C_1=\{1,17\},C_{-1}=\{31,15\},C_3=\{3,19\},C_{-3}=\{29,13\},C_{3^2}=\{9,25\},C_{-3^2}=\{23,7\},C_{3^3}=\{27,11\},C_{-3^3}=\{5,21\}$. From these cyclotomic cosets, we know that $1,3,5,7$ are in $C_{\pm 1},C_{\pm 3},C_{\pm 3^3},C_{\pm 3^2}$ respectively. By inserting $2,4,6,8$ from $C_2,C_{2^2},C_{6},C_{2^3}$ respectively, we get at least eight consecutive numbers $1,2,\cdots,8$. By Lemma \ref{BCH}, the minimum distance $d$ of the cyclic code $C=\langle g(x)\rangle$ is at least 9, where 
	\begin{equation*}
		(x-\alpha^{17})(x-\alpha^{19})(x-\alpha^{21})(x-\alpha^{23})\prod_{i=i}^{4}(x-\alpha^{2j-1})\prod_{j=1}^{4}(x-\alpha^{2j})\big |g(x).
	\end{equation*}
	 
\end{example}
\begin{theorem}\label{result1}
	Suppose that $q=1+2^zc$, where $z=n-1$ and c is odd. For any nonnegative integers $c_l$ and $c_r$ with $c_l+c_r<2^n$, there exists a  $(c_l,c_r)-[[2^n+c_l+c_r,2^{n-2}-2\delta_1]]_q$ QSC which corrects at least up to $\lfloor\frac{2^{n-2}-1}{2}\rfloor$ bit errors and at least up to $2^{n-3}$ phase errors, where $\delta_1\le 2^{n-2}-2$.
\end{theorem}
\begin{proof}
Let $\hat{M}_S(x)$ and $\overline{\hat{M}_S(x)}$ be the same as defined in Lemma \ref{BCH_bound}. Then let
\begin{equation*}
	g_1(x)=\hat{M}_S(x)\prod_{j=1}^{2^{n-3}}(x-\alpha^{2j})\prod_{i=1}^{\delta_1}M_{s_{i}}(x)
\end{equation*}
and 
\begin{equation*}
	g_2(x)=\overline{\hat{M}_S(x)}\prod_{j=1}^{2^{n-3}}(x-\alpha^{2j})\prod_{i=1}^{\delta_1}M_{s_i}^{\upvarepsilon_i}(x),
\end{equation*}
  where $\upvarepsilon_i$ is either $0$ or $1$ for each $i$, $\alpha$ is given by Definition \ref{alpha_def}, and $C_{s_i}$ are some cyclotomic cosets that exclude $\hat{C_S}$, $C_{2j}$ and $C_{2^n-2j}$. By Lemmas \ref{gnerate_matrix_dual_containing} and \ref{construct_augmented_c}, we have cyclic codes $C_a=\langle g_1(x)\rangle$ and $C_b=\langle g_2(x)\rangle$ satisfying $C_a^\perp\subseteq C_a\subseteq C_b$. Besides, $f(x)=\frac{g_1(x)}{g_2(x)}$ contains an irreducible factor of $M_S(x)$. Thus, by Lemma \ref{ord(f(x))bound}, we have ${\rm ord}(f(x))=2^n$, which is the upper bound of the synchronization capabilities of this QSC. By Lemma \ref{BCH_bound}, we know that the minimum distance $d_1$ of the cyclic code $C_a$ generated by $g_1(x)$ is at least $2^{n-2}+1$ and the minimum distance $d_2$ of the cyclic code $C_b$ generated by $g_2(x)$ is at least $2^{n-2}$. Note ${\rm deg}(g_1(x))=2^{n-2}+2^{n-3}+\delta_1$. Then by Lemma \ref{quantum_construction}, one can get the conclusion.\qed
\end{proof}

\begin{example}
	Let $z=3, c=5, n=4$. Then $q=1+2^zc=41$, and the length of the QSC is 16. When $r=n$, we obtain 
	\begin{equation*}
		M_S(x)=(x-\alpha)(x-\alpha^3)(x-\alpha^5)(x-\alpha^7)(x-\alpha^9)(x-\alpha^{11})(x-\alpha^{13})(x-\alpha^{15}). 
	\end{equation*}
	  According to Lemma \ref{BCH_bound}, let
	\begin{equation*}
	\hat{M}_S(x)=(x-\alpha)(x-\alpha^3)(x-\alpha^9)(x-\alpha^{11}),
	\end{equation*}
	 and
	 \begin{equation*}
	 	\begin{aligned}
	 	C_a&=\langle(x-\alpha)(x-\alpha^3)(x-\alpha^9)(x-\alpha^{11})\prod_{j=1}^{2}(x-\alpha^{2j})(x-\alpha^6)\rangle,\\	 
    	C_b&=\langle(x-\alpha^{3})(x-\alpha^{11})\prod_{j=1}^{2}(x-\alpha^{2j})\rangle.
 		\end{aligned}
 \end{equation*}
   Then we have $C_a=[16,9,6]_{41}$ and $C_b=[16,12,4]_{41}$.  Furthermore, $f(x)$ and $M_S(x)$ have the common factor $(x-\alpha)(x-\alpha^9)$. It means that the synchronization capabilities of this QSC reaches the upper bound 16 by Lemma \ref{ord(f(x))bound}. By theorem \ref{result1}, we obtain it with specific parameters $(c_l,c_r)-[[16+c_l+c_r,2]]_{41}$ and $c_l+c_r<16$ which can correct at least up to 1 bit errors and at least up to 2 phase errors.
	
\end{example}

\section{Conclusion}
In this paper, we obtain a new class of quantum synchronizable codes of length $2^n$ over ${\rm F}_q$, whose best attainable misalignment tolerance can be achieved, where $q\equiv 1({\rm mod}\ 4)$. These QSCs also possess good error-correcting capability towards bit errors and phase errors. Besides, the minimum distance of the cyclic codes we obtained is at least $2^{n-2}+1$ and $2^{n-2}$ respectively, which means the corresponding QSCs can correct at least up to$\lfloor\frac{2^{n-2}-1}{2}\rfloor$ bit errors and at least up to $2^{n-3}$ phase errors. It would be interesting to construct QSCs of length $2^n$ with maximal tolerance against misalignment based on constacyclic codes.


%
%


\begin{thebibliography}{}
%
\bibitem{Shor P W}Shor, P.W.: Scheme for reducing decoherence in quantum computer memory. Phys. Rev. A \textbf{52}(4), R2493-R2496(1995)
\bibitem{Steane A M}Steane, A.: Multiple particle interference and quantum error correction. Proc. roy. soc. london Ser. a. \textbf{452}, 2551-2557(1996)
\bibitem{Calderbank S}Calderbank, A.R., Shor, P.W.: Good quantum error-correcting codes exist. Phys. Rev. A \textbf{54}(2), 1098-1106(1996)
\bibitem{Calderbank R}Calderbank, A.R., Rains, E.M., Shor, P.W., Sloane, N.J.A.: Quantum error correction via codes over GF(4). IEEE International Symposium on Information Theory \textbf{44}, 1369-1387(1998)
\bibitem{FujiBsf}Fujiwara, Y.: Block synchronization for quantum information. Phys. Rev. A \textbf{87}(2), 022344(2013)
\bibitem{FujiTW}Fujiwara, Y., Tonchev, V.D., Wong, T.W.H.: Algebraic techniques in designing quantum synchronizable codes. Phys. Rev. A \textbf{88}(1), 012318(2013)
\bibitem{XieYFu}Xie, Y., Yuan, J., Fujiwara, Y.: Quantum synchronizable codes from quadratic residue codes and their supercodes. Proc. IEEE Inf. Theory Workshop(ITW), 172-176(2014)
\bibitem{LiZL}Li, L., Zhu, S., Liu, L.: Quantum synchronizable codes from the cyclotomy of order four. IEEE Commun. Lett. \textbf{23}(1), 12-15(2019)
\bibitem{ShiQX}Shi, X., Yue, Q., Huang, X.: Quantum synchronizable codes from the Whiteman's generalized cyclotomy. Cryptogr. Commun. \textbf{13}(2020)
\bibitem{WangT}Wang, T., Yan, T., Sidorenko, V., Wang, X.: Quantum synchronizable codes from cyclotomic classes of order two over $Z_{2q}$. arXiv:2106.02470v2(2021) 
\bibitem{G.K. Bakshi}Gurmeet, K., Bakshi, Madhu, Raka: A class of constacyclic codes over a finite field. Finite Fields Appl. \textbf{18}, 362-377(2012)
\bibitem{Mac Willianms}MacWilliams, F.J., Sloane, N.J.A.: The Theory of Error-Correctiong Codes 2nd edn. Amsterdam, The Netherlands: North Holland(1978)
\bibitem{San Ling1}Ling, S., Xing, C.: Coding Theory. Cambridge University Press, New York(2004)
\bibitem{WuYueS}Wu, Y., Yue, Q., Fan, S.: Self-reciprocal and self-conjugate-reciprocal irreducible factors of $x^n-\lambda$ and their applications. Finite Fields Appl. \textbf{63}, 101648(2020)
\bibitem{codetable}Grassl, M.: Bounds on the minimum distance of linear codes and quantum codes. Online available at https://www.codetables.de(2007)

\end{thebibliography}


\end{document}